\documentclass[12pt]{article}%
\usepackage{amsmath}
\usepackage{amsfonts}
\usepackage{amssymb}
\usepackage{graphicx}
\usepackage[open,openlevel=1]{bookmark}
\usepackage[most]{tcolorbox}
\usepackage{alphabeta}
\usepackage{boxedminipage}
\usepackage{multirow}
\usepackage{algorithmicx}
\usepackage{algpseudocode}
\usepackage{caption}
\usepackage{listings}
\usepackage{color}%
\providecommand{\U}[1]{\protect \rule{.1in}{.1in}}
\newtheorem{theorem}{Theorem}[section]

\newtheorem{proposition}[theorem]{Proposition}

\newenvironment{proof}[1][Proof]{\noindent \textbf{#1.} }{\  \rule{0.5em}{0.5em}}
\begin{document}

\title{On the Nash Equilibria\\of a Simple Discounted Duel}
\author{Athanasios Kehagias}
\maketitle

\begin{abstract}
We formulate and study a two-player \emph{static duel} game as a
\emph{nonzero-sum discounted stochastic game.} Players $P_{1},P_{2}$ are
standing in place and, in each turn, one or both \emph{may} shoot at the
\textquotedblleft other\textquotedblright \ player. If $P_{n}$ shoots at
$P_{m}$ ($m\neq n$), either he hits and kills him (with probability $p_{n}$)
or he misses him and $P_{m}$ is unaffected (with probability $1-p_{n}$). The
process continues until at least one player dies; if nobody ever dies, the
game lasts an infinite number of turns. Each player receives unit payoff for
each turn in which he remains alive; no payoff is assigned to killing the
opponent. We show that the the always-shooting strategy is a NE\ but, in
addition, the game also possesses \textquotedblleft \emph{cooperative}%
\textquotedblright \ (i.e., non-shooting) Nash equilibria in both stationary
and nonstationary strategies. A certain similarity to the repeated Prisoner's
Dilemma is also noted and discussed.

\end{abstract}

\section{Introduction\label{sec01}}

In this paper we study a two-player \emph{static duel} game played in turns.
Players $P_{1},P_{2}$ are standing in place and, in each turn, one or both
\emph{may} shoot at the \textquotedblleft other\textquotedblright \ player. If
$P_{n}$ shoots at $P_{m}$ ($m\neq n$), either he hits and kills him or he
misses him and $P_{m}$ is unaffected; the respective probabilities are $p_{n}$
and $1-p_{n}$. The process continues until at least one player dies; it is
possible that nobody ever dies and the game lasts an infinite number of turns.
We formulate the above as a \emph{nonzero-sum discounted stochastic game.} The
game rules and the players' payoff function will be presented in the next section.

Little work has been done on the static duel. Actually, as far as we know, it
has only been studied as a preliminary step in the study of the
\textquotedblleft static \emph{truel}\textquotedblright, in which \emph{three}
stationary players shoot at each other. Early works on the static truel are
\cite{Gardner1966,Kinnaird1946,Larsen1948,Mosteller1987,Shubik1954} in which
the postulated game rules guarantee the existence of \emph{exactly one}
survivor (\textquotedblleft winner\textquotedblright).{} A more general
analysis appears in \cite{Knuth1973} which considers the possibility of
\textquotedblleft cooperation\textquotedblright \ between the players. This
idea is further studied in
\cite{Kilgour1971,Kilgour1975,Kilgour1977,Zeephongsekul1991}. Recent papers on
the truel include
\cite{Amengual2005a,Amengual2005b,Amengual2006,Bossert2002,Brams1997,Brams2001,Brams2003,Dorraki2019,Toral2005,Wegener2021,Xu2012}%
. \footnote{Let us also note the existence of an extensive literature on a
quite different type of duel games, which essentially are \emph{games of
timing} \cite{Barron2013,Dresher1961,Karlin1959}. However, this literature is
not relevant to the game studied in this paper.}

While the above papers focus on various forms of the static truel, we believe
that the static duel is interesting in its own right and\ has not received the
attention it deserves. In particular we will show that, under our formulation,
the static duel has a certain similarity to the \emph{repeated Prisoner's
Dilemma }and possesses \textquotedblleft \emph{cooperative}\textquotedblright%
\ Nash equilibria in \emph{nonstationary strategies}.

This paper is structured as follows. In Section \ref{sec02} we define the game
rigorously. In Section \ref{sec03} we introduce several stationary and
nonstationary strategies and compute their expected payoffs. In Section
\ref{sec04} we prove that certain pairs of the previously defined strategies
are Nash equilibria. In Section \ref{sec05} we discuss the obtained results
and the connection of the static duel to the repeated Prisoner's Dilemma.
Finally, in Section \ref{sec06} we summarize our results and propose some
future research directions.

\section{The Game\label{sec02}}

The game involves players $P_{1},P_{2}$ and proceeds at discrete time steps
(\emph{rounds}) $t\in \left \{  1,2,...\right \}  $. The \emph{state }at time $t$
is
\[
s\left(  t\right)  =\left(  s_{1}\left(  t\right)  ,s_{2}\left(  t\right)
\right)  \in S=\left \{  \left(  1,1\right)  ,\left(  1,0\right)  ,\left(
0,1\right)  ,\left(  0,0\right)  ,\left(  \tau,\tau \right)  \right \}  .
\]
For $n\in \left \{  1,2\right \}  $, $s_{n}\left(  t\right)  $ \noindent is
$P_{n}$'s \emph{state }at $t\in \left \{  0,1,2,...\right \}  $ and can be
\[%
\begin{array}
[c]{ll}%
1: & \text{when }P_{n}\text{ is alive,}\\
0: & \text{when }P_{n}\text{ dies in the current round,}\\
\tau: & \text{when one or both players have died in a previous round.}%
\end{array}
\]
$P_{n}$'s \emph{action} at $t\in \left \{  1,2,...\right \}  $ is $f_{n}\left(
t\right)  $, which can be $1$ ($P_{n}$ is shooting) or $0$ ($P_{n}$ is not
shooting). If $s_{n}\left(  t-1\right)  \neq1$, $P_{n}$ cannot shoot at $t$
and $f_{n}\left(  t\right)  $ must equal $0$; if $s_{n}\left(  t-1\right)  =1$
then $f_{n}\left(  t\right)  $ can be either $0$ or $1$. When $f_{n}\left(
t\right)  =1$, $s_{-n}\left(  t\right)  =0$ (i.e., $P_{-n}$ dies\footnote{In
the sequel we use the standard game theoretic notation by which $s_{-1}=s_{2}%
$, $s_{-2}=s_{1}$. The same notation is used for players, actions etc.}) with
probability $p_{n}\in \left(  0,1\right)  $ and $s_{-n}\left(  t\right)  =1$
with probability $1-p_{n}$. We set $f\left(  t\right)  =\left(  f_{1}\left(
t\right)  ,f_{2}\left(  t\right)  \right)  $ and $\mathbf{p}=\left(
p_{1},p_{2}\right)  $. Note that we have assumed that $p_{1}$, $p_{2}$ are
different from both zero and one.

The game starts at an initial state $s\left(  0\right)  $; obviously, the main
case of interest is $s\left(  0\right)  =\left(  1,1\right)  $. At times
$t\in \left \{  1,2,...\right \}  $ the players choose simultaneously \ the
actions $f_{1}\left(  t\right)  $, $f_{2}\left(  t\right)  $ and the game
moves to state $s\left(  t\right)  $ according to the conditional \emph{state
transition probability} $\Pr \left(  s\left(  t\right)  |s\left(  t-1\right)
,f\left(  t\right)  \right)  $. \noindent In Figure 1 we present the state
transition diagram, in which the action-dependent\ transition probabilities
are written next to the edges; it is easily verified that these probabilities
conform to the game rules. \noindent The figure shows that the game starting
at $\left(  1,1\right)  $, lasts an infinite number of rounds and two
possibilities exist.

\begin{enumerate}
\item The game always stays in $\left(  1,1\right)  $ (no player is ever killed).

\item At some $t^{\prime}$ the game moves to a state $s\in \left \{  \left(
1,0\right)  ,\left(  0,1\right)  ,\left(  0,0\right)  \right \}  $ (one or both
players are killed) and at $t^{\prime}+1$ the game moves to the \emph{terminal
}state $\left(  \tau,\tau \right)  $, where it stays for ever after.
\end{enumerate}

\bigskip

\begin{minipage}{0.95\textwidth}
\medskip
\begin{center}
\includegraphics[scale=0.8,trim={2cm 16cm 0 3cm},clip]{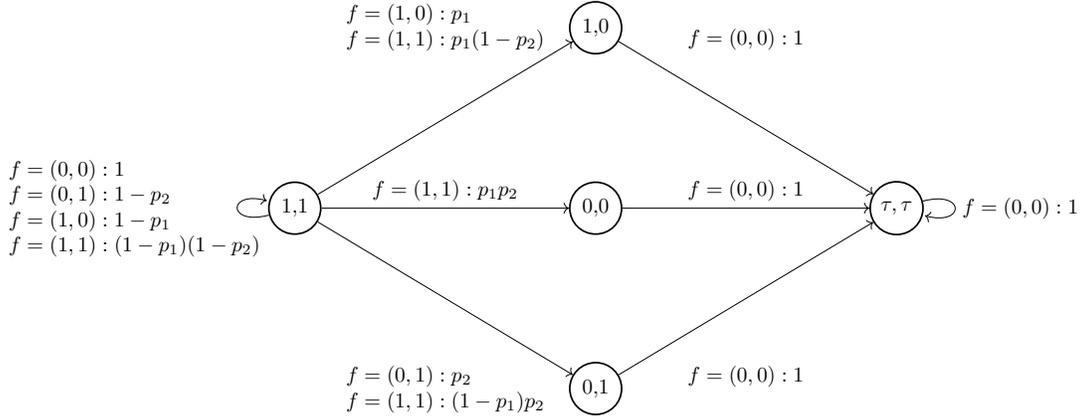}
\captionof{figure}{State transition diagram of the static duel game.}
\end{center}
\medskip
\end{minipage}

\bigskip

A \emph{finite history} is a sequence $h=s\left(  0\right)  f\left(  1\right)
s\left(  1\right)  ...f\left(  T\right)  s\left(  T\right)  $; an
\emph{infinite history} is an $h=s\left(  0\right)  f\left(  1\right)
s\left(  1\right)  ...\ $. An \emph{admissible} history is one which conforms
to the game rules; the set of all admissible finite (resp.
infinite)\ histories is denoted by $H$ (resp. $H^{\infty}$). For every history
$h$ and for $n\in \left \{  1,2\right \}  $, we define $P_{n}$'s \emph{total
payoff function} to be
\[
Q_{n}\left(  h\right)  =\sum_{t=0}^{\infty}\gamma^{t}q_{n}\left(  s\left(
t\right)  \right)
\]
where $\gamma \in \left(  0,1\right)  $ is the \emph{discounting factor } and
$q_{n}:S\rightarrow \mathbb{R}$ is $P_{n}$'s \emph{stage payoff function},
defined as follows.%
\[%
\begin{tabular}
[c]{ll}%
$q_{1}\left(  \tau,\tau \right)  =0,$ & $q_{2}\left(  \tau,\tau \right)  =0,$\\
$q_{1}\left(  1,1\right)  =1,$ & $q_{2}\left(  1,1\right)  =1,$\\
$q_{1}\left(  1,0\right)  =\frac{1}{1-\gamma},$ & $q_{2}\left(  1,0\right)
=0,$\\
$q_{1}\left(  0,1\right)  =0,$ & $q_{2}\left(  0,1\right)  =\frac{1}{1-\gamma
},$\\
$q_{1}\left(  0,0\right)  =0,$ & $q_{2}\left(  0,0\right)  =0.$%
\end{tabular}
\  \  \  \  \  \  \  \
\]
The above values indicate that each player receives one payoff unit for every
turn in which he stays alive; the payoff $q_{1}\left(  1,0\right)  =\frac
{1}{1-\gamma}$ incorporates the infinite payoff sequence $\sum_{t=0}^{\infty
}\gamma^{t}1=\frac{1}{1-\gamma}$ (this will result when $P_{1}$ kills $P_{2}$
and stays alive for the infinite number of subsequent turns). Note that a
player receives no direct payoff from killing his opponent, but he has the
indirect benefit of removing the possibility of being killed himself.

A \emph{strategy} for $P_{n}\ $is a function $\sigma_{n}:H\rightarrow \left[
0,1\right]  $ which corresponds to every finite history $h$ the probability
\[
x_{n}=\sigma_{n}\left(  h\right)  =\Pr \left(  \text{\textquotedblleft}%
P_{n}\text{ shoots at }P_{-n}\text{\textquotedblright}\right)  .
\]
A \emph{stationary strategy} is a $\sigma_{n}$ depending only on the current
state $s$, hence we simply write\ $x_{n}=\sigma_{n}\left(  s\right)  $. A
\emph{strategy profile} is a vector $\sigma=\left(  \sigma_{1},\sigma
_{2}\right)  $. We denote the set of all \emph{admissible} strategies (those
which are compatible with the game rules)\ by $\Sigma$ and the set of all
admissible stationary strategies by $\overline{\Sigma}$.

Given the initial state $s\left(  0\right)  $ and the strategies $\sigma_{1}$
and $\sigma_{2}$, used by $P_{1}$ and $P_{2}$ respectively, a probability
measure is defined on the set of all infinite histories. Since $\gamma
\in \left(  0,1\right)  $,\ the \emph{total expected payoffs}
\[
\forall n\in \left \{  1,2\right \}  :\overline{Q}_{n}\left(  s\left(  0\right)
,\sigma_{1},\sigma_{2}\right)  =\mathbb{E}\left(  Q_{n}\left(  h\right)
|s\left(  0\right)  ,\sigma_{1},\sigma_{2}\right)
\]
are well defined.

We have thus formulated the simultaneous static duel as a discounted
stochastic game, which we will denote by $\Gamma \left(  s\left(  0\right)
,\gamma,\mathbf{p}\right)  $ or simpl;y $\Gamma \left(  s\left(  0\right)
,\gamma \right)  $, when $\left(  p_{1},p_{2}\right)  $ is fixed. Our main
interest is in the \emph{nonzero-sum} game $\Gamma \left(  \left(  1,1\right)
,\gamma \right)  $. We assume that $P_{1}$ and $P_{2}$ attempt to reach a
\emph{Nash equilibrium} (NE), i.e., a strategy profile $\left(  \widehat
{\sigma}_{1},\widehat{\sigma}_{2}\right)  $ such that%
\[
\forall n\in \left \{  1,2\right \}  :\forall \sigma_{n}\in \Sigma:\overline{Q}%
_{n}\left(  \left(  1,1\right)  ,\widehat{\sigma}_{n},\widehat{\sigma}%
_{-n}\right)  \geq \overline{Q}_{n}\left(  \left(  1,1\right)  ,\sigma
_{n},\widehat{\sigma}_{-n}\right)  .
\]

\section{Some Basic Strategies and Their Payoffs\label{sec03}}

In this section we introduce several strategies which we will use in our later
exploration of Nash equilibria.

\subsection{The Stationary Strategy $\sigma^{S}$}

When $P_{n}$ ($n\in \left \{  1,2\right \}  $) uses a stationary admissible
strategy $\sigma_{n}$, we have $\sigma_{n}\left(  s\right)  =0$ for
$s\in \left \{  \left(  1,0\right)  ,\left(  0,1\right)  ,\left(  0,0\right)
,\left(  \tau,\tau \right)  \right \}  $ and $\sigma_{n}$ is fully specified by
the value $\sigma_{n}\left(  1,1\right)  =x_{n}$. Hence we will sometimes
write $\overline{Q}_{n}\left(  \left(  1,1\right)  ,x_{1},x_{2}\right)  $ in
place of $\overline{Q}_{n}\left(  \left(  1,1\right)  ,\sigma_{1},\sigma
_{2}\right)  $.

Let $V_{1}^{S}\left(  x_{1},x_{2}\right)  =\overline{Q}_{1}\left(  \left(
1,1\right)  ,x_{1},x_{2}\right)  $; for brevity we will also write simply
$V_{1}^{S}$. Then $V_{1}^{S}$ satisfies the equation%

\begin{align}
V_{1}^{S}  &  =1+\gamma p_{1}x_{1}\left(  x_{2}\left(  1-p_{2}\right)
+(1-x_{2})\right)  \frac{1}{1-\gamma}\label{eq0201}\\
&  \qquad \qquad \qquad+\gamma \left(  x_{1}\left(  1-p_{1}\right)  +\left(
1-x_{1}\right)  \right)  \left(  x_{2}\left(  1-p_{2}\right)  +\left(
1-x_{2}\right)  \right)  V_{1}^{S},\nonumber
\end{align}
obtained by the following reasoning. When the game is in state $s=\left(
1,1\right)  $, $P_{1}$'s expected payoff is one unit for the current state
plus the discounted expected payoff from the subsequent state $s^{\prime}$,
for which we have the following possibilities.

\begin{enumerate}
\item $s^{\prime}=\left(  1,0\right)  $ when $P_{1}$ shoots and hits $P_{2}$
and $P_{2}$ either shoots and misses or does not shoot; the respective
probability is $p_{1}x_{1}\left(  x_{2}\left(  1-p_{2}\right)  +(1-x_{2}%
)\right)  $. The total expected payoff of this case is $\overline{Q}%
_{1}\left(  \left(  1,0\right)  ,x_{1},x_{2}\right)  =\frac{1}{1-\gamma}$.

\item $s^{\prime}=\left(  1,1\right)  $ when \emph{each of} $P_{1}$ and
$P_{2}$ either shoots and misses or does not shoot; the respective probability
is $\left(  x_{1}\left(  1-p_{1}\right)  +\left(  1-x_{1}\right)  \right)
\left(  x_{2}\left(  1-p_{2}\right)  +\left(  1-x_{2}\right)  \right)  $. In
this case we have returned to the starting state $\left(  1,1\right)  $ and
the additional total expected payoff is again $\overline{Q}_{1}\left(  \left(
1,1\right)  ,x_{1},x_{2}\right)  $.

\item We also have the possibilities of moving into $\left(  0,1\right)  $ and
$\left(  0,0\right)  $, but these yield zero payoff to $P_{1}$, so they are
not included in (\ref{eq0201}).
\end{enumerate}

\noindent Solving (\ref{eq0201}) we get, after some algebraic
calculations\footnote{The calculations required to obtain the solution have
been performed by the computer algebra system Maple and then verified by hand.
This is also true of additional (sometimes quite complicated)\ calculations
required in the rest of the paper.}, that
\begin{equation}
V_{1}^{S}\left(  x_{1},x_{2}\right)  =\frac{1-\gamma \left(  1-p_{1}%
x_{1}\left(  1-p_{2}x_{2}\right)  \right)  }{\left(  1-\gamma \right)  \left(
1-\gamma \left(  1-p_{1}x_{1}\right)  \left(  1-p_{2}x_{2}\right)  \right)  }.
\label{eq0202}%
\end{equation}
In what follows we will often work with the \emph{normalized} total expected
payoff. In this case, it is
\begin{equation}
v_{1}^{S}\left(  x_{1},x_{2}\right)  =\left(  1-\gamma \right)  \overline
{Q}_{1}\left(  \left(  1,1\right)  ,x_{1},x_{2}\right)  =\frac{1-\gamma \left(
1-p_{1}x_{1}\left(  1-p_{2}x_{2}\right)  \right)  }{1-\gamma \left(
1-p_{1}x_{1}\right)  \left(  1-p_{2}x_{2}\right)  } \label{eq0203}%
\end{equation}
Formulas for $V_{2}^{S}\left(  x_{1},x_{2}\right)  =\overline{Q}_{2}\left(
\left(  1,1\right)  ,x_{1},x_{2}\right)  $ and $v_{2}^{S}\left(  x_{1}%
,x_{2}\right)  =\left(  1-\gamma \right)  \overline{Q}_{2}\left(  \left(
1,1\right)  ,x_{1},x_{2}\right)  $ can be obtained by interchanging the
indices $1$ and $2\ $in (\ref{eq0202})-(\ref{eq0203}).

\subsection{The Cooperating\ Strategy $\sigma^{C}$}

The stationary \textquotedblleft cooperating\textquotedblright \ (the name will
be justified in Section \ref{sec05})\ strategy $\sigma^{C}$ is defined by
\[
\sigma^{C}\left(  1,1\right)  =0\text{, which means the player never shoots.}%
\]
Obviously, $\sigma^{C}$ is $\sigma^{S}$ with $x_{n}=0$ and $\overline{Q}%
_{n}\left(  \left(  1,1\right)  ,\sigma^{C},\sigma^{C}\right)  =\overline
{Q}_{n}\left(  \left(  1,1\right)  ,0,0\right)  $. Hence we obtain $V_{1}%
^{C}=\overline{Q}_{1}\left(  \left(  1,1\right)  ,\sigma^{C},\sigma
^{C}\right)  =\overline{Q}_{1}\left(  \left(  1,1\right)  ,0,0\right)  $ by
setting $x_{1}=x_{2}=0$ in (\ref{eq0202}). Consequently the expected and
normalized expected total payoff are
\[
V_{1}^{C}=\frac{1}{1-\gamma},\qquad v_{1}^{C}=\left(  1-\gamma \right)
V_{1}^{C}=1
\]
Formulas for $V_{2}^{C}=\overline{Q}_{2}\left(  \sigma^{C},\sigma^{C}\right)
$, $v_{2}^{c}=\left(  1-\gamma \right)  V_{2}^{C}$ are obtained by exchanging
the indices $1$ and $2$ in the above formulas.

\subsection{The Defecting\ Strategy $\sigma^{D}$}

The stationary \textquotedblleft defecting\textquotedblright \ (the name will
be justified in Section \ref{sec05})\ strategy $\sigma^{D}$ is defined by
\[
\sigma^{D}\left(  1,1\right)  =1\text{, which means the player always shoots
with probability one.}%
\]
Obviously, $\sigma^{D}$ is $\sigma^{S}$ with $x_{n}=1$ and $\overline{Q}%
_{n}\left(  \left(  1,1\right)  ,\sigma^{D},\sigma^{D}\right)  =\overline
{Q}_{n}\left(  \left(  1,1\right)  ,1,1\right)  $. Hence we obtain $V_{1}%
^{D}=\overline{Q}_{1}\left(  \left(  1,1\right)  ,\sigma^{D},\sigma
^{D}\right)  =\overline{Q}_{1}\left(  \left(  1,1\right)  ,1,1\right)  $ by
setting $x_{1}=x_{2}=1$ in (\ref{eq0202}). We then get the expected and
normalized expected total payoff to be%
\begin{align*}
V_{1}^{D}  &  =\frac{1-\gamma \left(  1-p_{1}\left(  1-p_{2}\right)  \right)
}{\left(  1-\gamma \right)  \left(  1-\gamma \left(  1-p_{1}\right)  \left(
1-p_{2}\right)  \right)  }\\
v_{1}^{D}  &  =\left(  1-\gamma \right)  V_{1}^{D}\left(  \gamma \right)
=\frac{1-\gamma \left(  1-p_{1}\left(  1-p_{2}\right)  \right)  }%
{1-\gamma \left(  1-p_{1}\right)  \left(  1-p_{2}\right)  }%
\end{align*}
Formulas for $V_{2}^{D}=\overline{Q}_{2}\left(  \sigma^{D},\sigma^{D}\right)
$, $v_{2}^{D}=\left(  1-\gamma \right)  V_{2}^{D}$ are obtained by exchanging
the indices $1$ and $2$ in the above formulas.

\subsection{The Early-Shooting Strategy $\sigma^{DC\left(  K\right)  }$}

The nonstationary \textquotedblleft early-shooting\textquotedblright \ strategy
$\sigma^{DC\left(  K\right)  }$ is to shoot (with probability one)\ only at
times $1,2,...,K$, where $K$ is a parameter of the strategy. Let
$V_{1}^{DC\left(  K\right)  }=Q_{1}\left(  \left(  1,1\right)  ,\sigma
^{DC\left(  K\right)  },\sigma^{DC\left(  K\right)  }\right)  $; then
$V_{1}^{DC\left(  K\right)  }$ satisfies the equation:%
\begin{align}
V_{1}^{DC\left(  K\right)  }  &  =1\nonumber \\
&  +\gamma \left(  p_{1}\left(  1-p_{2}\right)  \frac{1}{1-\gamma}+\left(
1-p_{1}\right)  \left(  1-p_{2}\right)  \right) \nonumber \\
&  +\gamma^{2}\left(  \left(  1-p_{1}\right)  p_{1}\left(  1-p_{2}\right)
^{2}\frac{1}{1-\gamma}+\left(  1-p_{1}\right)  ^{2}\left(  1-p_{2}\right)
^{2}\right) \nonumber \\
&  +...\label{eq0204}\\
&  +\gamma^{K-1}\left(  \left(  1-p_{1}\right)  ^{K-2}p_{1}\left(
1-p_{2}\right)  ^{K-1}\frac{1}{1-\gamma}+\left(  1-p_{1}\right)  ^{K-1}\left(
1-p_{2}\right)  ^{K-1}\right) \nonumber \\
&  +\gamma^{K}\left(  \left(  1-p_{1}\right)  ^{K-1}p_{1}\left(
1-p_{2}\right)  ^{K}\frac{1}{1-\gamma}+\left(  1-p_{1}\right)  ^{K}\left(
1-p_{2}\right)  ^{K}V_{1}^{C}\right) \nonumber
\end{align}
This equation is justified as follows.

\begin{enumerate}
\item At time $t=0$, $P_{1}$ receives a payoff of one unit.

\item At times $t=1,...,K-1$, the expected payoffs are the following.

\begin{enumerate}
\item With probability $\left(  1-p_{1}\right)  ^{t-1}p_{1}\left(
1-p_{2}\right)  ^{t}$, $P_{1}$ misses at times $t^{\prime}=1,...,t-1$ and
succeeds at time $t^{\prime}$, while $P_{2}$ misses as times $t^{\prime
}=1,...,t$. In this case $P_{1}$ receives payoff $\frac{1}{1-\gamma}$.

\item With probability $\left(  1-p_{1}\right)  ^{t}\left(  1-p_{2}\right)
^{t}$, both $P_{1}$ and $P_{2}$ miss at times $t^{\prime}=1,...,t$. In this
case $P_{1}$ receives payoff $1$.

\item All other possibilities yield zero payoff, so they are not included in
the equation.
\end{enumerate}

\item At time $t=K$, the expected payoffs are the following.

\begin{enumerate}
\item With probability $\left(  1-p_{1}\right)  ^{K-1}p_{1}\left(
1-p_{2}\right)  ^{K}$, $P_{1}$ misses at times $t^{\prime}=1,...,K-1$ and
succeeds at time $t^{\prime}=K$, while $P_{2}$ misses as times $t^{\prime
}=1,...,K$. In this case $P_{1}$ receives payoff $\frac{1}{1-\gamma}$.

\item With probability $\left(  1-p_{1}\right)  ^{K}\left(  1-p_{2}\right)
^{K}$, both $P_{1}$ and $P_{2}$ miss at times $t^{\prime}=1,...,K$. In this
case both players will never shoot at subsequent times, so we have returned to
the starting state $\left(  1,1\right)  $ and the additional total expected
payoff is $V_{1}^{C}=\overline{Q}_{1}\left(  \left(  1,1\right)  ,\sigma
^{C},\sigma^{C}\right)  $. $\ $

\item All other possibilities yield zero payoff, so they are not included in
the equation.
\end{enumerate}
\end{enumerate}

\noindent Substituting in (\ref{eq0204})\ the expression for $V_{1}^{C}$ we
obtain the following expressions for expected total and normalized expected
total payoff:
\begin{align*}
V_{1}^{DC\left(  K\right)  }  &  =\frac{1+\gamma^{K+1}(1-p_{1})^{K}%
(1-p_{2})^{K}p_{2}-\gamma(1-p_{1}+p_{1}p_{2})}{\left(  1-\gamma \right)
\left(  1-\gamma \left(  1-p_{1}\right)  \left(  1-p_{2}\right)  \right)  }\\
v_{1}^{DC\left(  K\right)  }  &  =\left(  1-\gamma \right)  V_{1}^{DC\left(
K\right)  }=\frac{1+\gamma^{K+1}(1-p_{1})^{K}(1-p_{2})^{K}p_{2}-\gamma
(1-p_{1}+p_{1}p_{2})}{1-\gamma \left(  1-p_{1}\right)  \left(  1-p_{2}\right)
}%
\end{align*}
Formulas for $V_{2}^{DC\left(  K\right)  }$, $v_{2}^{DC\left(  K\right)  }$
are obtained by exchanging the indices $1$ and $2$ in the above formulas.

\subsection{The Late-Shooting Strategy $\sigma^{CD\left(  K\right)  }$}

The nonstationary \textquotedblleft late-shooting\textquotedblright \ strategy
$\sigma^{CD\left(  K\right)  }$ is to shoot (with probability one)\ only at
times $K,K+1,...$, where $K$ is a parameter of the strategy. Let
$V_{1}^{CD\left(  K\right)  }=Q_{1}\left(  \left(  1,1\right)  ,\sigma
^{CD\left(  K\right)  },\sigma^{CD\left(  K\right)  }\right)  $; then
$V_{1}^{CD\left(  K\right)  }$ satisfies the equation:%
\begin{equation}
V_{1}^{CD\left(  K\right)  }=1+\gamma+...+\gamma^{K-1}+\gamma^{K}\left(
p_{1}\left(  1-p_{2}\right)  \frac{1}{1-\gamma}+\left(  1-p_{1}\right)
\left(  1-p_{2}\right)  V_{1}^{D}\right)  \label{eq0205}%
\end{equation}
This equation is justified as follows.

\begin{enumerate}
\item At times $t=0,1,...,K-1$, $P_{1}$ receives discounted payoff of one unit.

\item At time $t=K$ we have the following possibilities.

\begin{enumerate}
\item With probability $p_{1}\left(  1-p_{2}\right)  $, $P_{1}$ hits $P_{2}$,
while $P_{2}$ misses. In this case $P_{1}$ receives payoff $\frac{1}{1-\gamma
}$.

\item With probability $\left(  1-p_{1}\right)  \left(  1-p_{2}\right)  $,
both $P_{1}$ and $P_{2}$ miss. In this case both $P_{1}$ and $P_{2}$ revert to
strategy $\sigma^{D}$ and $P_{1}$ receives total expected payoff $V_{1}%
^{D}=\overline{Q}_{1}\left(  \left(  1,1\right)  ,\sigma^{D},\sigma
^{D}\right)  $.

\item All other possibilities yield zero payoff, so they are not included in
the equation.
\end{enumerate}
\end{enumerate}

\noindent Substituting in (\ref{eq0205}) the previously obtained expression
for $V_{1}^{D}$ we get the following expressions for expected total and
normalized expected total payoff:
\begin{align*}
V_{1}^{CD\left(  K\right)  }  &  =\frac{1-\gamma^{K}p_{2}-\gamma
(1-p_{1})\left(  1-p_{2}\right)  }{\left(  1-\gamma \right)  \left(
1-\gamma \left(  1-p_{1}\right)  \left(  1-p_{2}\right)  \right)  }\\
v_{1}^{CD\left(  K\right)  }  &  =\left(  1-\gamma \right)  V_{1}^{CD\left(
K\right)  }=\frac{1-\gamma^{K}p_{2}-\gamma(1-p_{1})\left(  1-p_{2}\right)
}{1-\gamma \left(  1-p_{1}\right)  \left(  1-p_{2}\right)  }%
\end{align*}
Formulas for $V_{2}^{CD\left(  K\right)  }$, $v_{2}^{CD\left(  K\right)  }$
are obtained by exchanging the indices $1$ and $2$ in the above formulas.

\subsection{The Periodic-Shooting Strategy $\sigma^{P\left(  M\right)  }$}

The nonstationary \textquotedblleft periodic-shooting\textquotedblright%
\ strategy $\sigma^{P\left(  M\right)  }$ is to shoot only at times
$M+1,2M+2,...$, where $M$ is strategy parameter. Let $V_{1}^{P\left(
M\right)  }=Q_{1}\left(  \left(  1,1\right)  ,\sigma^{M},\sigma^{M}\right)  $;
by reasoning similar to that of the previous cases, we see that $V_{1}%
^{P\left(  M\right)  }$ satisfies the equation%
\begin{align}
V_{1}^{P\left(  M\right)  }  &  =1+\gamma+...+\gamma^{M}\nonumber \\
&  +\gamma^{M+1}\left(  p_{1}\left(  1-p_{2}\right)  \frac{1}{1-\gamma
}+\left(  1-p_{1}\right)  \left(  1-p_{2}\right)  V_{1}^{P\left(  M\right)
}\right)  \label{eq0206}%
\end{align}
Solving (\ref{eq0206})\ we get the following expressions for expected total
and normalized expected total payoff:
\begin{align*}
V_{1}^{P\left(  M\right)  }  &  =\frac{1-\gamma^{M+1}\left(  1-p_{1}\left(
1-p_{2}\right)  )\right)  }{\left(  1-\gamma \right)  \left(  1-\gamma
^{M+1}\left(  1-p_{1}\right)  \left(  1-p_{2}\right)  \right)  }\\
v_{1}^{P\left(  M\right)  }  &  =\frac{1-\gamma^{M+1}\left(  1-p_{1}\left(
1-p_{2}\right)  \right)  }{1-\gamma^{M+1}\left(  1-p_{1}\right)  \left(
1-p_{2}\right)  }%
\end{align*}
Formulas for $V_{2}^{P\left(  M\right)  }$, $v_{2}^{P\left(  M\right)  }$ are
obtained by exchanging the indices $1$ and $2$ in the above formulas.

\subsection{Grim Strategies}

Any $P_{n}$'s strategy $\sigma_{n}$ can be used to define a corresponding
\emph{grim} strategy $\widetilde{\sigma}_{n}$ as follows:%
\[%
\begin{array}
[c]{ll}%
\widetilde{\sigma}_{n}: & \text{As long as }P_{-n}\text{ uses }\sigma
_{n}\text{, }P_{n}\text{ also uses }\sigma_{n}\text{;}\\
& \text{if at the }t\text{-th turn }P_{-n}\text{ deviates from }\sigma
_{n}\text{, }P_{n}\text{\ uses }\sigma^{D}\text{ in all subsequent turns.}%
\end{array}
\]
For example, the grim-cooperating strategy $\widetilde{\sigma}^{C}$ dictates
that: $P_{n}$ never shoots, as long as $P_{-n}$ does not shoot either; if at
the $t$-th turn \thinspace$P_{-n}$ shoots then, starting at the $\left(
t+1\right)  $-th turn, $P_{n}$ will always shoot with probability one.
Similarly we can get:

\begin{enumerate}
\item the grim-defecting strategy $\widetilde{\sigma}^{D}$ (it is identical to
$\sigma^{D}$),

\item the grim-early-shooting strategy $\widetilde{\sigma}^{DC\left(
K\right)  }$,

\item the grim-late-shooting strategy $\widetilde{\sigma}^{CD\left(  K\right)
}$,

\item the grim-periodic-shooting strategy $\widetilde{\sigma}^{P\left(
M\right)  }$.
\end{enumerate}

\section{Nash Equilibria\label{sec04}}

We will now pesent a sequence of propositions; each one indicates that a
certain strategy pair is a (stationary or nonstationary)\ NE; sometimes this
will only hold for a certain range of $\gamma$ and possibly $p_{1},p_{2}$ values.

\begin{proposition}
\normalfont For every $\gamma \in \left(  0,1\right)  $, the only stationary
NE\ of $\Gamma \left(  \left(  1,1\right)  ,\gamma \right)  $ are $\left(
\sigma^{C},\sigma^{C}\right)  $ and $\left(  \sigma^{D},\sigma^{D}\right)  $.
\end{proposition}

\begin{proof}
Suppose that $P_{1}$ (resp. $P_{2}$) uses the stationary strategy $\sigma^{S}$
with $\sigma^{S}\left(  1,1\right)  =x_{1}$ (resp. $\sigma^{S}$ with
$\sigma^{S}\left(  1,1\right)  =x_{2}$ ). Then $P_{1}$'s payoff is
\[
V_{1}^{S}\left(  x_{1},x_{2}\right)  =\frac{1-\gamma \left(  1-p_{1}%
x_{1}\left(  1-p_{2}x_{2}\right)  \right)  }{\left(  1-\gamma \right)  \left(
1-\gamma \left(  1-p_{1}x_{1}\right)  \left(  1-p_{2}x_{2}\right)  \right)  }.
\]
Now suppose $P_{1}$ switches to $\sigma_{S}\left(  1,1\right)  =y_{1}$; his
payoff becomes
\[
V_{1}^{S}\left(  y_{1},x_{2}\right)  =\frac{1-\gamma \left(  1-p_{1}%
y_{1}\left(  1-p_{2}x_{2}\right)  \right)  }{\left(  1-\gamma \right)  \left(
1-\gamma \left(  1-p_{1}y_{1}\right)  \left(  1-p_{2}x_{2}\right)  \right)  }.
\]
Let us look at the difference of normalized payoffs%
\begin{align*}
\delta v_{1}  &  =\left(  1-\gamma \right)  \left(  Q_{1}\left(  \left(
1,1\right)  ,x_{1},x_{2}\right)  -Q_{1}\left(  \left(  1,1\right)
,y_{1},x_{2}\right)  \right) \\
&  =\frac{1-\gamma \left(  1-p_{1}x_{1}\left(  1-p_{2}x_{2}\right)  \right)
}{1-\gamma \left(  1-p_{1}x_{1}\right)  \left(  1-p_{2}x_{2}\right)  }%
-\frac{1-\gamma \left(  1-p_{1}y_{1}\left(  1-p_{2}x_{2}\right)  \right)
}{1-\gamma \left(  1-p_{1}y_{1}\right)  \left(  1-p_{2}x_{2}\right)  }\\
&  =\frac{\gamma^{2}p_{1}p_{2}x_{2}(x_{1}-y_{1})(1-p_{2}x_{2})}{((1-\gamma
(1-p_{2}x_{2})(1-p_{1}y_{1}))(1-\gamma(1-p_{2}x_{2})(1-p_{1}x_{1})))}%
\end{align*}
Now, $P_{1}$ has no incentive to switch from $x_{1}$ to $y_{1}$ iff $\delta
v_{1}\geq0$ which is equivalent to%
\[
\gamma^{2}p_{1}p_{2}x_{2}(x_{1}-y_{1})(1-p_{2}x_{2})\geq0
\]
Similarly, $P_{2}$ has no incentive to switch from $x_{2}$ to $y_{2}$ iff
\[
\gamma^{2}p_{1}p_{2}x_{1}(x_{2}-y_{2})(1-p_{1}x_{1})\geq0
\]
Hence, the following hold for $n\in \left \{  1,2\right \}  $.

\begin{enumerate}
\item If $\left(  x_{1},x_{2}\right)  =\left(  0,0\right)  $, $P_{n}$ has no
incentive to change $x_{n}$; $\left(  x_{1},x_{2}\right)  =\left(  0,0\right)
$; hence $\left(  \sigma^{C},\sigma^{C}\right)  $ is a NE.

\item If $\left(  x_{1},x_{2}\right)  =\left(  1,1\right)  $, $P_{n}$ has no
incentive to change $x_{n}$; $\left(  x_{1},x_{2}\right)  =\left(  1,1\right)
$; hence $\left(  \sigma^{D},\sigma^{D}\right)  $ is a NE.

\item If $\left(  x_{1},x_{2}\right)  \in \left(  0,1\right)  \times \left(
0,1\right)  $ then it cannot be a NE, because $P_{n}$ has incentive to change
(unilaterally)\ from $x_{n}$ to $1$.
\end{enumerate}

\noindent This completes the proof.
\end{proof}

\bigskip

\noindent Now we will start looking at NE obtained from combinations of grim strategies.

\bigskip

\begin{proposition}
\normalfont For every $\gamma \in \left(  0,1\right)  $, $\left(  \widetilde
{\sigma}^{C},\widetilde{\sigma}^{C}\right)  $ is a NE\ of $\Gamma \left(
\left(  1,1\right)  ,\gamma \right)  $.
\end{proposition}

\begin{proof}
Suppose that both $P_{1}$ and $P_{2}$ use $\widetilde{\sigma}^{C}$. Then
$P_{1}$'s payoff is
\[
V_{1}^{C}=Q_{1}\left(  \widetilde{\sigma}^{C},\widetilde{\sigma}^{C}\right)
=Q_{1}\left(  \sigma^{C},\sigma^{C}\right)  =\frac{1}{1-\gamma}.
\]
Now suppose $P_{1}$ deviates from $\widetilde{\sigma}^{C}$. It suffices to
examine the case in which $P_{1}$ deviates at $t=1$; furthermore, after
$P_{1}$ deviates (i.e., starting at $t=2$) $P_{2}$ will switch to $\sigma^{D}$
and $P_{1}$ has no incentive to not shoot at any $t\geq2$. \footnote{This is a
consequence of the following fact, which we will often use in the remainder of
the paper. If $P_{n}$ starts using a stationary strategy $\sigma_{n}$ at some
time $t$, then $P_{-n}$'s best response is also a stationary strategy. This is
the case because, for a fixed stationary $\sigma_{n}$, $P_{-n}$ has to solve a
\emph{Markov Decision Process}, for which the optimal strategy is stationary.
For more details see \cite{Sobel1971}.} Hence $P_{1}$ is essentially using the
strategy $\sigma_{1}=\sigma^{D}$ and his total expected \ payoff will then be
\begin{align*}
V_{1}  &  =Q_{1}\left(  \sigma^{D},\widetilde{\sigma}^{C}\right)
=1+\gamma \left(  p_{1}\frac{1}{1-\gamma}+\left(  1-p_{1}\right)  \left(
1+\gamma Q_{1}\left(  \sigma^{D},\sigma^{D}\right)  \right)  \right) \\
&  =1+\gamma \left(  p_{1}\frac{1}{1-\gamma}+\left(  1-p_{1}\right)  \left(
1+\gamma \frac{1-\gamma \left(  1-p_{1}\left(  1-p_{2}\right)  \right)
}{\left(  1-\gamma \right)  \left(  1-\gamma \left(  1-p_{1}\right)  \left(
1-p_{2}\right)  \right)  }\right)  \right) \\
&  =\frac{1-p_{2}(1-p_{1})\gamma^{3}-(1-p_{2})(1-p_{1})\gamma}{\left(
1-\gamma \right)  \left(  1-\gamma \left(  1-p_{1}\right)  \left(
1-p_{2}\right)  \right)  }%
\end{align*}
Now
\begin{align*}
\left(  1-\gamma \right)  \left(  Q_{1}\left(  \widetilde{\sigma}%
^{C},\widetilde{\sigma}^{C}\right)  -Q_{1}\left(  \sigma^{D},\widetilde
{\sigma}^{C}\right)  \right)   &  =\left(  1-\gamma \right)  \left(  V_{1}%
^{C}-V_{1}\right) \\
&  =1-\frac{1-p_{2}(1-p_{1})\gamma^{3}-(1-p_{2})(1-p_{1})\gamma}%
{1-\gamma \left(  1-p_{1}\right)  \left(  1-p_{2}\right)  }\\
&  =\allowbreak \frac{\gamma^{3}p_{2}\left(  1-p_{1}\right)  }{1-\gamma \left(
1-p_{1}\right)  \left(  1-p_{2}\right)  }>0
\end{align*}
Hence $P_{1}$ has no incentive to deviate from $\widetilde{\sigma}^{C}$. The
same can be proved for $P_{2}$. Consequently $\left(  \widetilde{\sigma}%
^{C},\widetilde{\sigma}^{C}\right)  $ is a NE.
\end{proof}

\bigskip

\noindent In the next proposition the strategy profile is a NE only for
\textquotedblleft large enough\textquotedblright \ $\gamma$.

\bigskip

\begin{proposition}
\normalfont There exist some $\gamma_{0}\in \left(  0,1\right)  $ such
that:$\ $for all $\gamma \in \left(  \gamma_{0},1\right)  $, and for all
$K\in \mathbb{N}$, $\left(  \widetilde{\sigma}^{DC\left(  K\right)
},\widetilde{\sigma}^{DC\left(  K\right)  }\right)  \ $is a NE of
$\Gamma \left(  \left(  1,1\right)  ,\gamma \right)  $.
\end{proposition}

\begin{proof}
Recall that, when both players use $\widetilde{\sigma}^{DC\left(  K\right)  }%
$, $P_{1}$ receives payoff%
\[
V_{1}^{DC\left(  K\right)  }=Q_{1}\left(  \left(  1,1\right)  ,\widetilde
{\sigma}^{DC\left(  K\right)  },\widetilde{\sigma}^{DC\left(  K\right)
}\right)  =\frac{1+\gamma^{K+1}(1-p_{1})^{K}(1-p_{2})^{K}p_{2}-\gamma
(1-p_{1}+p_{1}p_{2})}{\left(  1-\gamma \right)  \left(  1-\gamma \left(
1-p_{1}\right)  \left(  1-p_{2}\right)  \right)  }%
\]
Let us show that $P_{1}$ has no incentive to use a deviating strategy
$\sigma_{1}$.

\begin{enumerate}
\item \noindent \underline{\textbf{Case I:}} Let us first consider strategies
which deviate at times $t\in \left \{  K+1,K+2,...\right \}  $; i.e., they shoot
after the game has entered the no-shooting phase. We actually need to consider
only $\sigma_{1}$ which will shoot\ at $t=K+1$ and with probability one. In
this case
\begin{align*}
Q_{1}\left(  \left(  1,1\right)  ,\widetilde{\sigma}^{DC\left(  K\right)
},\widetilde{\sigma}^{DC\left(  K\right)  }\right)   &  =A+\gamma^{K+1}%
Q_{1}\left(  \left(  1,1\right)  ,\widetilde{\sigma}^{DC\left(  K\right)
},\widetilde{\sigma}^{DC\left(  K\right)  }\right) \\
&  =A+\gamma^{K+1}Q_{1}\left(  \left(  1,1\right)  ,\sigma^{C},\sigma
^{C}\right)  =A+\gamma^{K+1}V_{1}^{C}\\
Q_{1}\left(  \left(  1,1\right)  ,\sigma^{1},\widetilde{\sigma}^{DC\left(
K\right)  }\right)   &  =A+\gamma^{K+1}Q_{1}\left(  \left(  1,1\right)
,\sigma^{1},\widetilde{\sigma}^{DC\left(  K\right)  }\right) \\
&  =A+\gamma^{K+1}Q_{1}\left(  \left(  1,1\right)  ,\sigma^{1},\widetilde
{\sigma}^{DC\left(  K\right)  }\right)
\end{align*}
where $A$ is the expected payoff summed over times $t\in \left \{
0,...,K\right \}  $ and is the same for both strategies used by $P_{1}$. Now,
for the usual reasons, $P_{1}$ will keep shooting at $t\in \left \{
K+2,K+3,...\right \}  $ and we will have%
\begin{align*}
V_{1}  &  =Q_{1}\left(  \left(  1,1\right)  ,\sigma^{1},\widetilde{\sigma
}^{DC\left(  K\right)  }\right)  =p_{1}\frac{1}{1-\gamma}+\left(
1-p_{1}\right)  \left(  1+\gamma Q_{1}\left(  \sigma^{D},\sigma^{D}\right)
\right) \\
&  =p_{1}\frac{1}{1-\gamma}+\left(  1-p_{1}\right)  \left(  1+\gamma
\frac{1-\gamma \left(  1-p_{1}\left(  1-p_{2}\right)  \right)  }{\left(
1-\gamma \right)  \left(  1-\gamma \left(  1-p_{1}\right)  \left(
1-p_{2}\right)  \right)  }\right) \\
&  =\frac{1-p_{2}(1-p_{1})\gamma^{2}-(1-p_{2})(1-p_{1})\gamma}{\left(
1-\gamma \right)  \left(  1-\gamma \left(  1-p_{1}\right)  \left(
1-p_{2}\right)  \right)  }%
\end{align*}
Then we have
\begin{align*}
&  \left(  1-\gamma \right)  \left(  Q_{1}\left(  \left(  1,1\right)
,\widetilde{\sigma}^{C},\widetilde{\sigma}^{C}\right)  -Q_{1}\left(  \left(
1,1\right)  ,\sigma^{D},\widetilde{\sigma}^{C}\right)  \right) \\
&  =\left(  1-\gamma \right)  \left(  V_{1}^{C}-V_{1}\right) \\
&  =1-\frac{1-p_{2}(1-p_{1})\gamma^{2}-(1-p_{2})(1-p_{1})\gamma}%
{1-\gamma \left(  1-p_{1}\right)  \left(  1-p_{2}\right)  }\\
&  =\frac{p_{2}(1-p_{1})\gamma^{2}}{1-\gamma \left(  1-p_{1}\right)  \left(
1-p_{2}\right)  }>0
\end{align*}
Hence $P_{1}$ has no incentive to shoot at $t>K$.

\item \noindent \underline{\textbf{Case II:}} Let us next consider strategies
which deviate at times $t\in \left \{  1,2,...,K\right \}  $, i.e., they do not
shoot during the shooting phase. Again, after the first deviation $P_{1}$ has
no incentive to not shoot. So we only need to consider strategies $\sigma_{1}$
which (a)\ do not shoot at some $t=L\in \left \{  1,2,...,K\right \}  $ and
(b)\ shoot at all $t\in \left \{  1,2,...,L-1,L+1,...\right \}  $. Then, by the
usual arguments,
\begin{align*}
V_{1}  &  =Q_{1}\left(  \left(  1,1\right)  ,\sigma_{1},\widetilde{\sigma
}^{DC\left(  K\right)  }\right) \\
&  =1\\
&  +\gamma \left(  p_{1}\left(  1-p_{2}\right)  \frac{1}{1-\gamma}+\left(
1-p_{1}\right)  \left(  1-p_{2}\right)  \right) \\
&  +\gamma^{2}\left(  \left(  1-p_{1}\right)  p_{1}\left(  1-p_{2}\right)
^{2}\frac{1}{1-\gamma}+\left(  1-p_{1}\right)  ^{2}\left(  1-p_{2}\right)
^{2}\right) \\
&  +...\\
&  +\gamma^{L-1}\left(  \left(  1-p_{1}\right)  ^{L-2}p_{1}\left(
1-p_{2}\right)  ^{L-1}\frac{1}{1-\gamma}+\left(  1-p_{1}\right)  ^{L-1}\left(
1-p_{2}\right)  ^{L-1}\right) \\
&  +\gamma^{L}\left(  \left(  1-p_{1}\right)  ^{L-1}p_{1}\left(
1-p_{2}\right)  ^{L}\frac{1}{1-\gamma}+\left(  1-p_{1}\right)  ^{L}\left(
1-p_{2}\right)  ^{L}\left(  1+\gamma V_{1}^{D}\right)  \right) \\
&  =1+\sum_{k=1}^{L-1}\gamma^{k}\left(  \left(  1-p_{1}\right)  ^{k-1}%
p_{1}\left(  1-p_{2}\right)  ^{k}\frac{1}{1-\gamma}+\left(  1-p_{1}\right)
^{k}\left(  1-p_{2}\right)  ^{k}\right) \\
&  +\gamma^{L}\left(  1-p_{1}\right)  ^{L-1}\left(  1-p_{2}\right)  ^{L}%
p_{1}\frac{1}{1-\gamma}\\
&  +\gamma^{L}\left(  1-p_{1}\right)  ^{L}\left(  1-p_{2}\right)  ^{L}\left(
1+\gamma \frac{1-\gamma \left(  1-p_{1}\left(  1-p_{2}\right)  \right)
}{\left(  1-\gamma \right)  \left(  1-\gamma \left(  1-p_{1}\right)  \left(
1-p_{2}\right)  \right)  }\right)
\end{align*}
Let $\ $%
\begin{align*}
\delta v_{1}\left(  \gamma \right)   &  =\left(  1-\gamma \right)  \left(
Q_{1}\left(  \left(  1,1\right)  ,\widetilde{\sigma}^{DC\left(  K\right)
},\widetilde{\sigma}^{DC\left(  K\right)  }\right)  -Q_{1}\left(  \left(
1,1\right)  ,\sigma_{1},\widetilde{\sigma}^{DC\left(  K\right)  }\right)
\right) \\
&  =V_{1}^{DC\left(  K\right)  }-V_{1}.
\end{align*}
Note that $\delta v_{1}\left(  \gamma \right)  $ is well defined and continuous
for all $\gamma \in \left[  0,1\right]  $, because the factor $\left(
1-\gamma \right)  $ cancels the $\left(  1-\gamma \right)  $ factor in the
denominator of
\[
Q_{1}\left(  \left(  1,1\right)  ,\widetilde{\sigma}^{DC\left(  K\right)
},\widetilde{\sigma}^{DC\left(  K\right)  }\right)  -Q_{1}\left(  \left(
1,1\right)  ,\sigma_{1},\widetilde{\sigma}^{DC\left(  K\right)  }\right)  .
\]
After a considerable amount of algebra\footnote{Using Maple once again.} we
find that
\[
\delta v_{1}\left(  1\right)  =\frac{\left(  1-p_{1}\right)  ^{K}\left(
1-p_{2}\right)  ^{K}p_{2}}{p_{1}+\left(  1-p_{1}\right)  p_{2}}>0.
\]
Since $\delta v_{1}\left(  \gamma \right)  $ is continuous, there will exist
some $\gamma_{0}\in \left(  0,1\right)  $ such that $\delta v_{1}\left(
\gamma \right)  $ will be positive for every $\gamma \in \left(  \gamma
_{0},1\right)  $ and for every $K\in \mathbb{N}$. Hence, for such values,
$P_{1}$ has no incentive to deviate during the shooting phase.
\end{enumerate}

\noindent Putting together Cases I\ and II\ we see that $P_{1}$ has no
incentive to deviate from $\widetilde{\sigma}^{DC\left(  K\right)  }$. The
same is proved, similarly, for $P_{2}$. Hence $\left(  \widetilde{\sigma
}^{DC\left(  K\right)  },\widetilde{\sigma}^{DC\left(  K\right)  }\right)  $
is a NE.
\end{proof}

\bigskip

\noindent Next we present a negative result:\ mutual late shooting is not a NE.

\bigskip

\begin{proposition}
\normalfont For every $\gamma \in \left(  0,1\right)  $ and every $K\in
\mathbb{N}$, $\left(  \widetilde{\sigma}^{CD\left(  K\right)  },\widetilde
{\sigma}^{CD\left(  K\right)  }\right)  $ is not a NE of $\Gamma \left(
\left(  1,1\right)  ,\gamma \right)  $.
\end{proposition}

\begin{proof}
Recall that
\[
V_{1}^{CD\left(  K\right)  }=Q_{1}\left(  \widetilde{\sigma}^{CD\left(
K\right)  },\widetilde{\sigma}^{CD\left(  K\right)  }\right)  =\frac
{1-\gamma^{K}p_{2}-\gamma(1-p_{1})\left(  1-p_{2}\right)  }{\left(
1-\gamma \right)  \left(  1-\gamma \left(  1-p_{1}\right)  \left(
1-p_{2}\right)  \right)  }%
\]
We just need to show that $P_{1}$ has one profitable deviating strategy
$\sigma_{1}$. Let $\sigma_{1}$ be:\  \ do not shoot at $t\in \left \{
1,2,...,K-2\right \}  $, shoot at $t\in \left \{  K-1,K,...\right \}  $; in other
words start shooting one turn before the shooting phase starts. Then, by the
usual arguments, $P_{1}$'s payoff is%
\begin{align*}
V_{1}  &  =Q_{1}\left(  \left(  1,1\right)  ,\sigma_{1},\widetilde{\sigma
}^{CD\left(  K\right)  }\right)  =\sum_{k=0}^{K-2}\gamma^{k}+\gamma
^{K-1}\left(  p_{1}\frac{1}{1-\gamma}+\left(  1-p_{1}\right)  \left(  1+\gamma
V_{D}\right)  \right) \\
&  =\frac{\left(  -\gamma^{K+1}p_{2}(1-p_{1})+1-\gamma(1-p_{2})(1-p_{1}%
)\right)  }{\left(  1-\gamma \right)  (1-\gamma(1-p_{2})(1-p_{1}))}%
\end{align*}
By appropriate substitutions and algebraic calculations, we get
\begin{align*}
\delta v_{1}  &  =\left(  1-\gamma \right)  \left(  Q_{1}\left(  \widetilde
{\sigma}^{CD\left(  K\right)  },\widetilde{\sigma}^{CD\left(  K\right)
}\right)  -Q_{1}\left(  \sigma_{1},\widetilde{\sigma}^{CD\left(  K\right)
}\right)  \right) \\
&  =-\frac{p_{2}\gamma^{K}\left(  1-\gamma \left(  1-p_{1}\right)  \right)
}{(1-\gamma(1-p_{2})(1-p_{1}))}<0
\end{align*}
Hence $P_{1}$ has incentive to switch to $\sigma_{1}$ and $\left(
\widetilde{\sigma}^{CD\left(  K\right)  },\widetilde{\sigma}^{CD\left(
K\right)  }\right)  $ is not a NE.
\end{proof}

\begin{proposition}
\normalfont For every $M\in \mathbb{N}$, there exists a $\delta_{M}>0$ such
that: if
\begin{align*}
\gamma_{M}  &  =\frac{9}{10},\\
p_{M}  &  =\frac{1-e^{-M}}{10},\\
I_{M}  &  =\left(  \gamma_{M}-\delta_{M},\gamma_{M}+\delta_{M}\right)
\times \left(  p_{M}-\delta_{M},p_{M}+\delta_{M}\right)  \times \left(
p_{M}-\delta_{M},p_{M}+\delta_{M}\right)  ,
\end{align*}
then $\left(  \widetilde{\sigma}^{P\left(  M\right)  },\widetilde{\sigma
}^{P\left(  M\right)  }\right)  $ is a NE of $\Gamma \left(  \left(
1,1\right)  ,\gamma,\mathbf{p}\right)  $ for every $\left(  \gamma,p_{1}%
,p_{2}\right)  \in I_{M}$.
\end{proposition}

\begin{proof}
Recall that
\[
V^{P\left(  M\right)  }=Q_{1}\left(  \widetilde{\sigma}^{P\left(  M\right)
},\widetilde{\sigma}^{P\left(  M\right)  }\right)  \text{ }=\frac
{1-\gamma^{M+1}\left(  1-p_{1}\left(  1-p_{2}\right)  )\right)  }{\left(
1-\gamma \right)  \left(  1-\gamma^{M+1}\left(  1-p_{1}\right)  \left(
1-p_{2}\right)  \right)  }%
\]
We will prove that, for every $\left(  \gamma,p_{1},p_{2}\right)  \in I_{M}$,
$P_{1}$ has no incentive to deviate from $\widetilde{\sigma}^{P\left(
M\right)  }$ (the proof for$P_{2}$ is identical).

Suppose that $P_{1}$ uses some strategy $\sigma_{1}$ by which he shoots at
$P_{2}$ at some $t\neq i\cdot \left(  M+1\right)  $. For the usual reasons, it
suffices to consider strategies by which $P_{1}$ shoots in the first period
and with probability one. So suppose that $P_{1}$ abstains for all
$t\in \left(  1,...,K\right)  $ and then shoots at $P_{2}$ at some $t^{\prime
}=K+1\leq M$. Then the following two possibilities exist.

\begin{enumerate}
\item With probability $p_{1}$: $P_{2}$ is killed and $P_{1}$ receives payoff
$\frac{1}{1-\gamma}$.

\item With probability $1-p_{1}$: $P_{2}$ is missed, $P_{1}$ receives payoff
one and for all subsequent rounds $P_{2}$ will always shoot at $P_{1}$ with
probability one. In this case $P_{1}$'s best response at time $t^{\prime
\prime}>t^{\prime}$ is to always shoot at $P_{2}$ with probability one; hence,
starting at the $\left(  t^{\prime}+1\right)  $-th round, both players use the
$\sigma^{D}$ strategy. The total expected payoff received by $P_{1}$ in this
case is $Q_{1}\left(  \sigma^{D},\sigma^{D}\right)  $.
\end{enumerate}

\noindent Hence, assuming $P_{1}$ will first shoot at $t=K+1\in \left \{
1,...,M\right \}  $, by the above reasoning $P_{1}$'s expected total payoff
will be%
\[
V_{1}=Q_{1}\left(  \sigma_{1},\widetilde{\sigma}^{P\left(  M\right)  }\right)
\text{ }=\left(  \sum_{k=0}^{K}\gamma^{k}\right)  +\gamma^{K+1}\left(
p_{1}\frac{1}{1-\gamma}+\left(  1-p_{1}\right)  Q_{1}\left(  \sigma^{D}%
,\sigma^{D}\right)  \right)
\]
Substituting the $Q_{1}\left(  \widetilde{\sigma}^{P\left(  M\right)
},\widetilde{\sigma}^{P\left(  M\right)  }\right)  $\ and $Q_{1}\left(
\sigma^{D},\sigma^{D}\right)  $ values and performing a considerable amount of
algebra we get%
\begin{align*}
\delta v_{1}\left(  \gamma,p_{1},p_{2}\right)   &  =\left(  1-\gamma \right)
\left(  Q_{1}\left(  \widetilde{\sigma}^{P\left(  M\right)  },\widetilde
{\sigma}^{P\left(  M\right)  }\right)  -Q_{1}\left(  \sigma_{1},\widetilde
{\sigma}^{P\left(  M\right)  }\right)  \text{ }\right) \\
&  =\frac{p_{2}\gamma^{K+2}\left(  -\left(  1-p_{1}\right)  ^{2}\left(
1-p_{2}\right)  {\gamma}^{M+1}+\left(  1-p_{1}\right)  \left(  1-p_{2}\right)
{\gamma}^{M-K}-{\gamma}^{M-K-1}+1-p_{1}\right)  }{\left(  1-\gamma
^{M+1}\left(  1-p_{1}\right)  \left(  1-p_{2}\right)  \right)  \left(
1-\gamma \left(  1-p_{1}\right)  \left(  1-p_{2}\right)  \right)  }%
\end{align*}
Setting $p_{1}=p_{2}=p$ we get
\[
\delta v_{1}\left(  \gamma,p,p\right)  =\frac{p\gamma^{K+2}\left(  -\left(
1-p\right)  ^{3}{\gamma}^{M+1}+\left(  1-p\right)  ^{2}{\gamma}^{M-K}-{\gamma
}^{M-K-1}+1-p\right)  }{\left(  1-\gamma^{M+1}\left(  1-p\right)  ^{2}\right)
\left(  1-\gamma \left(  1-p\right)  ^{2}\right)  }%
\]
The sign of $\delta v_{1}\left(  \gamma,p,p\right)  $ is the same as that of%
\begin{align*}
f_{M,K}\left(  \gamma,p\right)   &  =-\left(  1-p\right)  ^{3}{\gamma}%
^{M+K+3}+\left(  1-p\right)  ^{2}{\gamma}^{M+2}-{\gamma}^{M+1}+\left(
1-p\right)  \gamma^{K+2}\\
&  =f_{1,M,K}\left(  \gamma,p\right)  +f_{2,M,K}\left(  \gamma,p\right)
\end{align*}
with
\begin{align*}
f_{M,K,1}\left(  \gamma,p\right)   &  =\left(  1-p\right)  ^{2}{\gamma}%
^{M+2}-\left(  1-p\right)  ^{3}{\gamma}^{M+K+3}\\
f_{M,K,2}\left(  \gamma,p\right)   &  =-{\gamma}^{M+1}+\left(  1-p\right)
\gamma^{K+2}%
\end{align*}
Now we consider the following cases.

\begin{enumerate}
\item \noindent \underline{\textbf{Case I: }$K\leq M-2$}. Then $M-K-1\geq1$.
For all $M$ and $K\in \left \{  1,...,M-2\right \}  $ we have
\[
\left(  1-p\right)  ^{2}>\left(  1-p\right)  ^{3}\text{ and }{\gamma}%
^{M+2}>{\gamma}^{M+K+3}%
\]
hence we will always have $f_{1,M,K}\left(  \gamma,p\right)  >0$. To also have
$f_{M,K,2}\left(  \gamma,p\right)  >0$ for a specific $K$, it suffices that
\[
{\gamma}^{M-K-1}<1-p\Leftrightarrow \gamma<\left(  1-p\right)  ^{\frac
{1}{M-K-1}}%
\]
To have $f_{M,K,2}\left(  \gamma,p\right)  >0$ for \emph{all} $K\in \left \{
1,...,M-2\right \}  $, it suffices that
\begin{equation}
\gamma<\left(  1-p\right)  ^{\frac{1}{M-2}} \label{eq0301}%
\end{equation}
In other words, for all $M$ we have:
\begin{align*}
\gamma &  \in \left(  0,\left(  1-p\right)  ^{\frac{1}{M-2}}\right)
\Rightarrow \left(  \forall K\in \left \{  1,...,M-2\right \}  :f_{M,K}\left(
\gamma,p\right)  >0\right) \\
\gamma &  \in \left(  0,\left(  1-p\right)  ^{\frac{1}{M-2}}\right)
\Rightarrow \left(  \forall K\in \left \{  1,...,M-2\right \}  :\delta
v_{1}\left(  \gamma,p,p\right)  >0\right)
\end{align*}

\item \noindent \underline{\textbf{Case II: }$K=M-1$}. In this case we want
\begin{align*}
-\left(  1-p\right)  ^{3}{\gamma}^{M+K+3}+\left(  1-p\right)  ^{2}{\gamma
}^{M+2}-{\gamma}^{M+1}+\left(  1-p\right)  \gamma^{K+2}  &  >0\Leftrightarrow
\\
-\left(  1-p\right)  ^{3}{\gamma}^{M+M-1+3}+\left(  1-p\right)  ^{2}{\gamma
}^{M+2}-{\gamma}^{M+1}+\left(  1-p\right)  \gamma^{M-1+2}  &
>0\Leftrightarrow \\
-\left(  1-p\right)  ^{3}{\gamma}^{2M+2}+\left(  1-p\right)  ^{2}{\gamma
}^{M+2}-{\gamma}^{M+1}+\left(  1-p\right)  \gamma^{M+1}  &  >0\Leftrightarrow
\\
-\left(  1-p\right)  ^{3}{\gamma}^{M+1}+\left(  1-p\right)  ^{2}{\gamma}%
-{1}+\left(  1-p\right)   &  >0\Leftrightarrow \\
-\left(  1-p\right)  ^{3}{\gamma}^{M+1}+\left(  1-p\right)  ^{2}{\gamma}-p  &
>0\Leftrightarrow \\
-\left(  1-p\right)  ^{3}{\gamma}^{M+1}+{\gamma}p^{2}-\left(  2\gamma
+1\right)  p+\gamma &  >0
\end{align*}
Let us define the function
\[
\overline{f}_{M}\left(  \gamma,p\right)  =-\left(  1-p\right)  ^{3}{\gamma
}^{M+1}+{\gamma}p^{2}-\left(  2\gamma+1\right)  p+\gamma
\]
By continuity, in a sufficiently small neighborhood of $\left(  \gamma
_{M},p_{M}\right)  =\left(  \frac{9}{10},\frac{1-e^{-M}}{10}\right)  $,
$\ $the$\ $sign of $\overline{f}_{M}\left(  \gamma,p\right)  $ will be the
same as that of%
\[
h_{1}\left(  M\right)  =\overline{f}_{M}\left(  \gamma_{M},p_{M}\right)
=-\frac{\left(  9+e^{-M}\right)  ^{3}\left(  \frac{9}{10}\right)  ^{M+1}%
}{1000}+\frac{9\left(  1-e^{-M}\right)  ^{2}}{1000}+\frac{31}{50}%
+\frac{7e^{-M}}{25}%
\]
and it suffices to show that $h_{1}\left(  M\right)  >0$ for all $M$. To this
end we first note that
\[
h_{1}\left(  1\right)  =-\frac{\left(  9+e^{-1}\right)  ^{3}\left(  \frac
{9}{10}\right)  ^{2}}{1000}+\frac{9\left(  1-e^{-1}\right)  ^{2}}{1000}%
+\frac{31}{50}+\frac{7e^{-1}}{25}=\allowbreak0.6\allowbreak0703...>0
\]
Also, letting%
\[
h_{2}\left(  M\right)  =-\frac{\left(  9+e^{-M}\right)  ^{3}\left(  \frac
{9}{10}\right)  ^{M+1}}{1000}+\frac{31}{50}%
\]
we have
\[
\forall M:h_{1}\left(  M\right)  >h_{2}\left(  M\right)  .
\]
Now, $h_{2}\left(  M\right)  $ is strictly increasing in $M$ and $h_{2}\left(
2\right)  =0.06422\,2...$ . Consequently%
\[
\forall M\geq2:h_{1}\left(  M\right)  >h_{2}\left(  M\right)  >h_{2}\left(
2\right)  >0
\]
Hence finally we have
\[
\forall M\geq1:\overline{f}_{M}\left(  \gamma_{M},p_{M}\right)  =h_{1}\left(
M\right)  >0.
\]

\end{enumerate}

\noindent Now, to have%
\[
\forall M,\forall K\in \left \{  1,...,M-1\right \}  :f_{M,K}\left(  \gamma
_{M},p_{M}\right)  >0
\]
we must ensure that (\ref{eq0301}) holds for $\left(  \gamma,p\right)  =$
$\left(  \gamma_{M},p_{M}\right)  $. In other words, we want $\gamma
_{M}<\left(  1-p_{M}\right)  ^{\frac{1}{M-2}}$ or, equivalently,
\[
\frac{9}{10}<\left(  1-\frac{1-e^{-M}}{10}\right)  ^{\frac{1}{M-2}}=\left(
\frac{9}{10}+\frac{e^{-M}}{10}\right)  ^{\frac{1}{M-2}}.
\]
This holds: since for all $M\in \mathbb{N}$ we have $\frac{9}{10}+\frac{e^{-M}%
}{10}<1$, we also have
\[
\frac{9}{10}<\left(  \frac{9}{10}+\frac{e^{-M}}{10}\right)  <\left(  \frac
{9}{10}+\frac{e^{-M}}{10}\right)  ^{\frac{1}{M-2}}.
\]

\noindent In short we have shown that
\begin{align*}
\forall M,\forall K  &  \in \left \{  1,...,M-1\right \}  :f_{M,K}\left(
\gamma_{M},p_{M}\right)  >0\\
\forall M,\forall K  &  \in \left \{  1,...,M-1\right \}  :\delta v_{1}\left(
\gamma_{M},p_{M},p_{M}\right)  >0
\end{align*}
For all $M$ and $K$, $\delta v_{1}\left(  \gamma,p_{1},p_{2}\right)  $ is a
continuous function. Hence, for all $M$, there exists some $\delta_{M}>0$ such
that
\[
\forall K\in \left \{  1,...,M-1\right \}  ,\forall \left(  \gamma,p_{1}%
,p_{2}\right)  \in I_{M}:\delta v_{1}\left(  \gamma,p_{1},p_{2}\right)  >0
\]
which shows that $P_{1}$ has no incentive to deviate from $\widetilde{\sigma
}^{P\left(  M\right)  }$. The same argument can be applied to $P_{2}$.
\ Hence, for every $\left(  \gamma,p_{1},p_{2}\right)  \in I_{M}$ , $\left(
\widetilde{\sigma}^{P\left(  M\right)  },\widetilde{\sigma}^{P\left(
M\right)  }\right)  $ is a NE of $\Gamma \left(  \left(  1,1\right)
,\gamma,\mathbf{p}\right)  $.
\end{proof}

\section{Some Additional Remarks\label{sec05}}

Let us now justify our terms \textquotedblleft cooperating\textquotedblright%
\ and \textquotedblleft defecting\textquotedblright \ strategy. From the
results of Section \ref{sec03}, for $n\in \left \{  1,2\right \}  $, we have
\begin{align*}
\overline{Q}_{n}\left(  \left(  1,1\right)  ,\sigma^{C},\sigma^{C}\right)   &
=\frac{1}{1-\gamma}\\
\overline{Q}_{n}\left(  \left(  1,1\right)  ,\sigma^{D},\sigma^{D}\right)   &
=\frac{1-\gamma \left(  1-p_{n}\left(  1-p_{-n}\right)  \right)  }{\left(
1-\gamma \right)  \left(  1-\gamma \left(  1-p_{1}\right)  \left(
1-p_{2}\right)  \right)  }.
\end{align*}
It follows that
\[
\overline{Q}_{n}\left(  \left(  1,1\right)  ,\sigma^{C},\sigma^{C}\right)
-\overline{Q}_{n}\left(  \left(  1,1\right)  ,\sigma^{D},\sigma^{D}\right)
=\frac{\gamma p_{-n}}{\left(  1-\gamma \right)  \left(  1-\gamma \left(
1-p_{1}\right)  \left(  1-p_{2}\right)  \right)  }>0.
\]
In short, just like in PD, it is more profitable for both players to not-shoot
rather than shoot. Because in our formulation there is no direct profit from
killing the opponent, both $\left(  \sigma^{C},\sigma^{C}\right)  $ and
$\left(  \sigma^{D},\sigma^{D}\right)  $ are NE; however, for both players,
$\left(  \sigma^{C},\sigma^{C}\right)  $ is more profitable NE\ than $\left(
\sigma^{D},\sigma^{D}\right)  $. This is the reason for calling $\left(
\sigma^{C},\sigma^{C}\right)  $ a cooperating, and $\left(  \sigma^{D}%
,\sigma^{D}\right)  $ a defecting strategy.

All this may be surprising, since one would expect that, in a duel, each
player's goal will be to eliminate his opponent. It may be supposed that the
higher profitability of $\left(  \sigma^{C},\sigma^{C}\right)  $ follows from
our choice of not assigning any direct payoff to killing one's opponent. But
this is not true. Even with a positive \textquotedblleft killing
payoff\textquotedblright, $\overline{Q}_{n}\left(  \left(  1,1\right)
,\sigma^{C},\sigma^{C}\right)  $ can still be greater than $\overline{Q}%
_{n}\left(  \left(  1,1\right)  ,\sigma^{D},\sigma^{D}\right)  $, provided
$\gamma$ is sufficiently close to one\footnote{This, as well as additional
results regarding the positive killing payoff case, will be reported in a
future publication.}. The reason for the superiority of $\left(  \sigma
^{C},\sigma^{C}\right)  $ is this:\ if a positive payoff is assigned to
survival, this, compounded over an infinite number of turns, can always
outweigh the killing payoff. Hence our model can be understood as a more
\textquotedblleft pacifist\textquotedblright \ version than the usual duel
model\footnote{This point has also been raised by Donald Knuth in the context
of the truel \cite{Knuth1973}. For example, he remarks that \textquotedblleft
a player who passes is guaranteeing that his opponent has no reason to shoot
back, as far as the opponent's survival is concerned \textquotedblright.}.

Let us now compare our static duel to the PD. In both the PD\ and the duel,
cooperation is more profitable than defection. While $\left(  \sigma
^{C},\sigma^{C}\right)  $ is not a NE in PD, $\left(  \sigma^{D},\sigma
^{D}\right)  $ is a NE in both of them. However, both the duel and the
repeated version of PD, possess several NE\ in grim strategies; the common
characteristic of all such equilibria is that they promote cooperation or, in
other words, punish defection (shooting). In fact, similarly to the case of
repeated PD, it might be possible to prove a \textquotedblleft Folk
Theorem\textquotedblright \ for the static duel as well; namely that every
feasible and individually rational payoff is a NE\ for $\gamma$ sufficiently
close to one. We intend to study this question in the future.

\section{Conclusion\label{sec06}}

We have formulated the simultaneous shooting static duel as a discounted
stochastic game. We have shown that it has two Nash equilibria in stationary
strategies, namely the \textquotedblleft always-shooting\textquotedblright%
\ and the \textquotedblleft never-shooting\textquotedblright \ strategies; in
addition several nonstationary, \textquotedblleft
cooperation-promoting\textquotedblright \ Nash equilibria also exist. \ In the
future we intend to extend the study of the static duel in several directions.

First, we want to extend our study and obtain similar results for two
variants:\ (a)\ the case of non-zero killing payoff and (b)\ the case of
terminal-only payoffs. In addition, we want to formulate and study a version
of the static duel in which each player wants to kill his opponent in the
\emph{shortest possible time}.

Secondly, we hope to prove a form of \textquotedblleft \emph{Folk
Theorem}\textquotedblright, namely that every every feasible and individually
rational payoff is a NE\ for $\gamma$ sufficiently close to one.

Finally, we want to formulate the static \emph{Nuel} (i.e., the duel-like game
which involves $N$ players shooting at each other) as a discounted stochastic
game and extend our results for this case.

\end{document}